\documentclass[10pt,a4paper]{article}
\usepackage{jheppub}
\usepackage[utf8]{inputenc}
\usepackage[english]{babel}
\usepackage{amsmath}
\usepackage{amsfonts}
\usepackage{amssymb}
\usepackage{amsthm}
\usepackage{float}
\usepackage{graphicx}
\usepackage{comment}
\graphicspath{ {./} }

\newcommand{\bb}[1]{\boldsymbol{#1}}
\theoremstyle{plain}
\newtheorem{Remark}{Remark}
\newtheorem{Theorem}{Theorem}
\newtheorem{Lemma}{Lemma}
\newtheorem{Proposition}{Proposition}
\newtheorem{Corollary}{Corollary}
\newtheorem{Definition}{Definition}

\graphicspath{{./Figures/}}
\DeclareMathOperator\arctanh{arctanh}

\title{PDE/statistical mechanics duality: relation between Guerra's interpolated $p$-spin ferromagnets and the Burgers hierarchy}
\author[a,b,c]{Alberto Fachechi}
\affiliation[a]{Dipartimento di Matematica e Fisica ``Ennio De Giorgi'', Universit\`a del Salento, Italy}
\affiliation[b]{GNFM-INdAM, Gruppo Nazionale di Fisica Matematica (Istituto Nazionale di Alta Matematica), Italy}
\affiliation[c]{INFN, Istituto Nazionale di Fisica Nucleare, Italy}
\emailAdd{alberto.fachechi@le.infn.it}
\abstract{We examine the duality relating the equilibrium dynamics of the mean-field $p$-spin ferromagnets at finite size in the Guerra's interpolation scheme and the Burgers hierarchy. In particular, we prove that - for fixed $p$ - the expectation value of the order parameter on the first side w.r.t. the generalized partition function satisfies the $p-1$-th element in the aforementioned class of nonlinear equations. In the light of this duality, we interpret the phase transitions in the thermodynamic limit of the statistical mechanics model with the development of shock waves in the PDE side. We also obtain the solutions for the $p$-spin ferromagnets at fixed $N$, allowing us to easily generate specific solutions of the corresponding equation in the Burgers hierarchy. Finally, we obtain an effective description of the finite $N$ equilibrium dynamics of the $p=2$ model with some standard tools in PDE side.}
\keywords{$p$-spin, Statistical mechanics, Burgers hierarchy, Nonlinear systems, Mean-field theory, PDE}
\begin{document}
\maketitle

\section{Introduction}

During the last decades, Statistical Mechanics of disordered systems has acquired a prominent role in describing complex phenomena and emerging properties of system with highly non-trivial dynamics. In particular, one of the major success is the possibility to analyze the relaxation and the equilibrium dynamics of spin-glass models \cite{MPV,SK}, {\it i.e.} systems of simple {degrees of freedom} (the spins, adopting the physics jargon) whose interaction are frustrated (or, in other words, {spins} in the system can compete with each others) and known only through their probability distribution. {Random and frustrated interactions} are responsible for a rich phenomenology ({\it e.g.} the existence of multiple timescales for the relaxation toward the equilibrium), which is due to the complex structure of the free-energy landscape (see for example \cite{MPV}). The analysis of spin-glass equilibrium dynamics is commonly accomplished by computing the so-called {\it quenched} free-energy\footnote{With quenched free-energy, we mean {the expectation value w.r.t. to the intrinsic disorder (the couplings) is performed after taking the logarithm, {\it i.e.} the intensive free energy is $f_N^{Q}(\beta)=-(\beta N)^{-1}\mathbb E \log Z_N(\beta)$, to be compared with its {\it annealed} counterpart, {\it i.e. }$f_N^{A}(\beta)=-(\beta N)^{-1} \log\mathbb E Z_N(\beta)$. Thus, the disorder in the two versions of the free energy is treated in a substantial different way. The relevant quantity for spin-glass models is the former, since it deals with the equilibrium dynamics of the system at fixed couplings realization.}} in the thermodynamic limit (whose analysis is motivated by its self-averaging property \cite{guerraton2,guerra3}) in terms of few order parameters (in particular, the so-called Edward-Anderson overlap \cite{ea}) and tracing the phase diagram in the parameter space. Remarkably, Statistical Mechanics of disordered systems has gained a renewed interest with the advent of Deep Learning \cite{lecun}, since it offers the ideal tool for investigating information processing in neural networks (see for example \cite{PhysRevLett,Agliari_2019,ags,BARRA2018205,Guerra2,Coolen,offequilibrium1,FACHECHI201924,Hopfield2554,Personnaz,KanterSompo}). In particular, Statistical Mechanics commonly deals with the Machine Retrieval regime, {\it i.e.} the long-term relaxation of neural networks which are fed with some inputs to be classified according to some previously learnt features. Focusing on the methods, the solution of such complex systems can be found by means of the straightforward (but unfortunately non-rigorous) {\it replica trick} approach. At the same time, rigorous and mathematically transparent approaches can be developed, see for example \cite{contucci1,gg,guerra1,guerraton,panchenko1,panchenko2,panchenko3,Tala2,Tala1,Tala}. For our concerns, the most important one is the Guerra's interpolation scheme, which makes use of sum rules for the computation of relevant quantities \cite{Agliari_2019,Agliari_2019_1,AGLIARI2020254,Alemanno_2020,guerra2,barrareplica} or taking benefit of mechanical analogy or relation with PDEs \cite{linda,pspin,BARRA2018205,barrapde1,Barra-Guerra-HJ,barrapde2}. In the latter approach, the quenched free-energy can be interpolated in an effective mean field scheme, in which the interpolation parameters can be interpreted as coordinates of a fictious spacetime $(t,x)$. In this case, it is possible to use the entire set of PDE technology in order to solve the thermodynamics of the model, also for finite $N$. Even if such interpolation methods were initially developed in the context of spin-glass, they are quite general, and by this they can be easily brought to similar (and possibily simpler) models, such as $p$-spin ferromagnets. The latters are spin systems in which the {spins} interact {with coupling of order $p$} and cooperatively ({\it i.e.} they are not spin-glasses). In this paper, we study the relation of Guerra interpolated partition functions of $p$-spin ferromagnets and prove that the expectation values of the order parameter (the global magnetization) w.r.t. to the associated Boltzmann-Gibbs measure is in 1-1 correspondence with the elements of the Burgers hierarchy \cite{KUDRYASHOV20091293,olver,Sharma,tasso1976cole}. This duality can be explored in both sense: for instance, we can generate specific solutions of the Burgers hierarchy by exploiting the finite size solution of the $p$-spin systems; on the other hand, we can achieve informations about the thermodynamics of the ferromagnets from the PDE side.
\par\medskip
Despite being interesting by itself, this duality could be a promising tool for investigating the principles of information processing in biological phenomena. Indeed, as firstly suggested in \cite{thompson}, Statistical Mechanics turned out to be an effective tool in the description of universal phenomena in biological systems, see also \cite{Hopfield625,Amit,aglmultitasking,Butler11833,Mora5405,Agliari_2013,STANLEY1994214,PhysRevLett1}. For instance, kinetics of biochemical reactions can be framed in a purely statistical mechanics picture in terms of the Curie-Weiss ferromagnetic model, see \cite{cwkinetic1,cwkinetic2}. However, since in standard statistical mechanics one considers the thermodynamic limit $N\to \infty$, this scenario is a good approximation only for systems with large size. As opposite to this picture, recently the research in biochemical information processing has focused on phenomena involving system with small size \cite{felix}, for which the statistical mechanics scenario loses its predictive power. Thus, for such practical application one has to consider the thermodynamics at finite size $N$. On the other side, it has been established \cite{multisite} that the binding of {spins} in systems such as long chain molecules (such as {\it proteins}) can takes place with multisite interactions ({\it i.e.} $p>2$) beyond the pairwise scenario, see also \cite{dibiasio}. In this picture, our duality can thus provide rigorous mathematical tools to investigate these interesting biological situations.
\par\medskip
The paper is organized as follows. In Section \ref{sec:2}, we provide the basic tools of statistical mechanics for $p$-spin ferromagnetic models, also discussing the thermodynamic solution of the free-energy (both with physical and mathematically rigorous methods, {\it i.e.} Guerra's interpolation scheme). In Section \ref{sec:3}, we provide the fundamental notions about the Burgers hierarchy and the reduction to linear PDEs through Cole-Hopf transforms. In Section \ref{sec:4}, we prove the duality between the two sides and some results (such as the interpretation of phase transition in the ferromagnets as the development of shock waves in the PDE side and the search of solutions of the Burgers hierarchy by means of finite size solutions). We also give a description of $p=2$ ferromagnetic model with the tools offered by the Burgers equation.

\section{{A cursory look at mean-field $p$-spin ferromagnetic models}}\label{sec:2}
In this Section, we will introduce the fundamental notions about the $p$-spin ferromagnetic models. In particular, we will discuss the solution of the model, both with purely physics arguments and Guerra's interpolation schemes. We will also discuss the existence of the thermodynamic limit for the free-energy of the models. Let us start by introducing the following

{
\begin{Definition}
Let $\bb{\sigma}\in \Omega_N =\{-1,+1\}^N$ be a general configuration of the system at finite size $N$, and let $J>0$ be the interaction strength. The Hamilton function of the $p$-spin ferromagnetic model is defined as
\begin{equation}
H_N (\bb \sigma\vert J)=-\frac J{N^{p-1}} \sum_{i_1,i_2,\dots,i_p=1}^N \sigma_{i_1}\sigma_{i_2}\dots \sigma_{i_p}.
\end{equation}
\end{Definition}
}
{
\begin{Remark}
The normalization factor $1/N^{p-1}$ is inserted in order to ensure the linear extensivity of the Hamilton function, meaning that the energy of the system scales linearly in its volume:
$$
H_N(\bb \sigma \vert J)= N \epsilon_N (\bb \sigma \vert J),
$$
where $\epsilon_N(\bb \sigma)$ (the energy per site of the model associated to the configuration $\bb\sigma$) is finite in the infinite-size limit $N\to\infty$.
\end{Remark}
}
{
\begin{Remark}
Traditionally, the Hamilton function of a $p$-spin ferromagnetic model is given by
$$
H'_N (\bb \sigma \vert J)=-\frac{J}{N^{p-1}}\sum_{1\le i_1<i_2 < \dots < i_p \le N}\sigma_{i_1}\sigma_{i_2}\dots \sigma_{i_p}.
$$
However, the difference between the two formulations is negligible in the thermodynamic limit. Indeed, by noticing that
$$
\sum_{i_1,i_2,\dots,i_p=1}^N \equiv p! \sum_{1\le i_1<i_2 < \dots < i_p \le N},
$$
holding in $N\to\infty$ limit (see for example \cite{crisanti_pspin}), it is easy to see that
$$
\frac{H'_N (\bb \sigma\vert J)}N= \frac{  H_N (\bb \sigma\vert J)}{p ! N}+\text{corrections vanishing at } N\to\infty,
$$
thus the thermodynamics of the systems within the two frameworks are equivalent (and they differ only up to a trivial rescaling of the temperature by a factor $p!$).
\end{Remark}
}

{
\begin{Definition}
Let $\beta \in \mathbb R^+$ the level of thermal noise (i.e. the inverse temperature $\beta=T^{-1}$).
Then, the partition function of the ferromagnetic $p$-spin model is defined as
\begin{equation}\label{eq:2}
Z_N(\beta, J)=\sum_{\bb{\sigma}}\exp\big(-\beta  H_N (\bb\sigma \vert J)\big)\equiv \sum_{\bb{\sigma}}\exp\Big(\frac{\beta J}{N^{p-1}} \sum_{i_1,i_2,\dots,i_p=1}^N\sigma_{i_1}\sigma_{i_2}\dots \sigma_{i_p}\Big),
\end{equation}
where $\sum_{\bb \sigma} \equiv \sum_{\bb\sigma\in\Sigma_N}$ is the sum over all possible configurations of the system. The Boltzmann factor corresponding to the partition function \eqref{eq:2} is defined 
\begin{equation}
\label{eq:Bweight}
B_N (\bb \sigma)=\exp\Big(\frac{\beta J}{N^{p-1}} \sum_{i_1,i_2,\dots,i_p=1}^N\sigma_{i_1}\sigma_{i_2}\dots \sigma_{i_p}\Big),
\end{equation}
where of course $Z_N (\beta,J)=\sum_{\bb \sigma}B_N (\bb\sigma)$.
\end{Definition}
}

\begin{Definition}
The global magnetization of the ferromagnetic $p$-spin model is defined as
\begin{equation}\label{eq:1}
m_N{(\bb{\sigma})}=\frac1N\sum_{i=1}^N\sigma_i.
\end{equation}
\end{Definition}

\begin{Remark}
Since the model is ferromagnetic, this is the only order parameters we need to fully describe the equilibrium of the model.
\end{Remark}

{
\begin{Remark}
The expression of the Hamilton function in terms of the global magnetization is
\begin{equation}
\label{eq:HandM}
H_N (\bb \sigma\vert J)= -J N \Big(\frac{1}{N}\sum_{i=1}^N \sigma_i\Big)^p =-J N m_N (\bb \sigma)^p.
\end{equation}
\end{Remark}
}

\begin{Remark}
Notice that, as usual, the ferromagnetic strength $J$ has the only effect of rescaling the thermal noise in the {system}. Therefore, without loss of generality, we can set $J=1$. In this way, we will denote {$H_N(\bb\sigma)\equiv H_N(\bb\sigma\vert J=1)$ and} $Z_N(\beta)\equiv Z_N(\beta,J=1)$.
\end{Remark}

{
\begin{Definition}
Given a function $F(\bb{\sigma})$ of the spins in the {system}, its expectation value is defined as
\begin{equation}\label{eq:3}
\omega[F(\bb {\sigma})]=\frac{\sum_{\bb{\sigma}}F(\bb {\sigma})B_N(\bb \sigma)}{Z_N (\beta)}.
\end{equation}
\end{Definition}
}

\begin{Remark}
We stress that, for odd $p$, the Hamilton function is not invariant under gauge transformations $\bb {\sigma}\to -\bb {\sigma}$. This means that, as opposite to the even $p$ cases, there are no gauge-equivalent solutions.\end{Remark}

\begin{Definition}
The intensive statistical pressure $A_N (\beta)$ of the system is defined as
\begin{equation}
\label{eq:5}
A_N (\beta)=\frac1N\log Z_N(\beta).
\end{equation}
\end{Definition}

{
\begin{Remark}
The intensive statistical pressure is related to the usual (intensive) Helmholtz free energy $f_N(\beta)$ as $A_N(\beta)=-\beta f_N(\beta)$. Let us denote with $P(\bb \sigma)$ the Boltzmann-Gibbs distribution of the system, which is of course defined as
$$
P (\bb \sigma)= \frac{1}{Z_N (\beta)}\exp(-\beta H_N (\bb \sigma)).
$$
Given the energy per site at fixed system configuration
$$
\epsilon_N(\bb \sigma)= \frac{H_N (\bb \sigma)}{N},
$$
and
$$
s_N (\bb \sigma)=\frac{\log P(\bb \sigma)}{N},
$$
we can write down the equality
\begin{equation}
\label{eq:centralequality}
A_N(\beta)=\omega[s_N(\bb {\sigma})]-\beta \omega [\epsilon _N(\bb {\sigma})],
\end{equation}
which, a part for a factor $-\beta$, is the relation between the (intensive) Helmholtz free energy and the expectation values of the relevant thermodynamic observables for the system. Indeed, the first contribution is nothing but the entropy per site, since, according to the definition \eqref{eq:3}, we have
$$
\omega[s_N (\bb \sigma)]= -\frac 1N \sum_{\bb \sigma} P(\bb\sigma)\log P (\bb\sigma).
$$
Clearly, $S_N [P]=-\sum_{\bb \sigma} P(\bb\sigma)\log P (\bb \sigma)$ is the Shannon entropy associated to the probability distribution $P(\bb\sigma)$. Thus, Eq. \eqref{eq:centralequality} exactly paints the relation between the intensive pressure $A_N(\beta)$ and the intensive Helmoltz free energy $f_N(\beta)$.
\end{Remark}
}

For our concerns, it is relevant the statistical pressure in the thermodynamic limit
\begin{equation}
\label{eq:6}
A(\beta)=\lim_{N\to \infty}\frac1N\log Z_N(\beta).
\end{equation}
The solution of the $p$-spin ferromagnetic models can be straightforwardly obtained by purely statistical tools. The {expression of the expectation value of the energy per site follows directly from \eqref{eq:HandM}}, and reads
$$
\omega [\epsilon _N(\bb {\sigma})]=-\omega [m_N { (\bb {\sigma})}^p].
$$
Regarding the entropy per site, due to the mean-field nature of the model we make the assumption that the equilibrium probability distribution to observe the system in a given configuration can be factorized as product of probabilities of independent sites $P(\bb {\sigma})=\prod_{i=1}^N P(\sigma_i)$, and the spin orientation is driven by the global magnetization. Then, in the thermodynamic limit we can write
\begin{equation}
\label{eq:meanfieldansatz}
P(\sigma_i)=\frac{1+\bar m}{2}\delta_{\sigma_i,1}+\frac{1-\bar m}{2}\delta_{\sigma_i,-1},
\end{equation}
where $\bar m$ is the thermodynamic value of the magnetization:
$$
\bar m = \lim_{N\to \infty}\omega[m_N{ (\bb {\sigma})}].
$$

{
\begin{Remark}
When working at finite size $N$, the computation of the expectation value of the energy per site requires evaluating correlation functions of the form $\omega[\sigma_{i_1}\sigma_{i_2}\dots \sigma_{i_p}]$, which is in general a non-trivial task. However, due to the mean-field nature of the ferromagnetic model, simplifications occur in the thermodynamic limit. Indeed, as remarked above, in this limit the probability distribution in the configuration space factorizes, {\it i.e.} $P(\bb \sigma)=\prod_i P(\sigma_i)$, so that the evaluation of the aforementioned correlation functions becomes trivial. Likewise, it is possible to achieve the same result by assuming the self-averaging property for the global magnetization. In simple words, we require the fluctuations of the order parameter w.r.t. its equilibrium value $\bar m$ to vanish as $N\to \infty$. This is translated in mathematical terms by requiring that the probability distribution of the global magnetization converges to a Dirac-delta distribution which is centered around the equilbrium value, i.e.
$$
\lim_{N\to \infty} P_N (m_N(\bb\sigma))= \delta (m-\bar m),
$$
 $\beta$ almost everywhere, with $m$ being the global magnetization in the thermodynamic limit. The self-averaging assumption is a reasonable hypothesis for ferromagnetic systems, see for example \cite{barrapspin}, as the number of pure states does not depend on the system size. As a consequence, the expectation value of a general function $F$ of the global magnetization can be simply evaluated:
$$
\omega [F(m (\bb \sigma))]\equiv \int dm_N P_N (m _N) F(m_N)\underset{N\to \infty}\to \int dm \ \delta (m-\bar m) F(m)= F(\bar m). 
$$
\end{Remark}
}
{
The previous remark implies that the expectation value of the energy per site in the thermodynamic limit can be evaluated as
$$
\lim_{N\to \infty} \omega [\epsilon_N(\bb\sigma)]= -\bar m ^p.
$$
Regarding the entropy contribution, the mean-field assumption \eqref{eq:meanfieldansatz} allows us to discard correlations between the spins and then reduces the problem to one-body computations. Then, with this expression of the probability distribution we can write down the entropy using the Shannon prescription:
$$
\omega[s_N (\bb \sigma)]=-\frac1N\sum_{\bb {\sigma}} P(\bb {\sigma })\log P(\bb {\sigma})=-\left(\frac{1+\bar m}2\log \frac{1+\bar m}2+\frac{1-\bar m}2\log \frac{1-\bar m}2 \right),
$$
which is valid in the thermodynamic limit. Then, putting everything together we get}
\begin{equation}
\label{eq:7}
A(\beta)=\beta \bar m^p-\frac{1+\bar m}2\log \frac{1+\bar m}2-\frac{1-\bar m}2\log \frac{1-\bar m}2.
\end{equation}
By imposing the extremality condition for the {intensive} pressure $\partial_{\bar m}A(\beta)=0$, we have
\begin{equation}
\label{eq:8}
\beta p \bar m^{p-1}-\arctanh (\bar m)=0\Rightarrow \bar m = \tanh (\beta p \bar m^{p-1}),
\end{equation}
{which is exactly the self-consistency equations for the $p$-spin ferromagnetic model \cite{barrapspin}.}

\subsection{Existence of the thermodynamic limit}

In the previous Section, it was tacitly supposed the existence of the thermodynamic limit for the $p$-spin ferromagnetic models as described by the partition function \eqref{eq:2}. For the sake of simplicity, we rigorously prove the existence of the thermodynamic limit for even $p$ (however, this results can be carried out also for models with Hamilton functions which are polynomial in the global magnetization $m_N {(\bb {\sigma}) }$, therefore including also the odd $p$ case, see for example \cite{contucci}). The key idea is to separate the ferromagnetic models with $N$ interacting spins in two distinct non-interacting subsystems respectively with $N_1$ and $N_2$ spins, such that $N=N_1+N_2$. Then, we introduce an interpolating partition function with the following

\begin{Definition}
Let $\sigma_i =\pm 1$ for $i=1,\dots,N$ be the binary spins of the model and $t\in [0,1]$ an interpolating parameter. Let us build two separate, non-interacting subsystems by collecting respectively the spins $\sigma_i$ with $i=1,\dots, N_1$ and $\sigma_i$ for $i =N_1+1,\dots N_1+N_2$, with corresponding order parameters
\begin{equation}
m_{1}{(\bb {\sigma}) }=\frac1{N_1}\sum_{i=1}^{N_1}\sigma_i, \quad m_{2}{(\bb {\sigma}) }=\frac1{N_2}\sum_{i={N_1+1}}^{N_1+N_2}\sigma_i.
\end{equation}
Then, the interpolating partition function is defined as 
\begin{equation}
\label{eq:2.1.1}
Z_N(\beta, t)= \sum_{ \bb {\sigma}}\exp \big(t N \beta m_N  {(\bb {\sigma}) }^p+(1-t)N_1 \beta m_{1}{(\bb {\sigma}) }^p+ (1-t)N_2\beta m_{2}{(\bb {\sigma}) }^p \big).
\end{equation}
{The Boltzmann factor associated to the partition function \eqref{eq:2.1.1} is
$$
B_N(\bb \sigma,t)=\exp \big(t N \beta m_N  {(\bb {\sigma}) }^p+(1-t)N_1 \beta m_{1}{(\bb {\sigma}) }^p+ (1-t)N_2\beta m_{2}{(\bb {\sigma}) }^p \big),
$$
so that $Z_N(\beta,t)=\sum_{\bb\sigma} B_{N}(\bb\sigma,t)$. The statistical pressure of the model is defined as
\begin{equation}
A_N (\beta,t)= \frac1N \log Z_N(\beta,t).
\end{equation}
}
\end{Definition}
{
	\begin{Remark}
		We stress that, within the Guerra's interpolating framework, we do not divide the system of size $N$ in two subsystems with resp. $N_1$ and $N_2$ spins. Indeed, in order for this decomposition to hold in the previous interpretation, one has to discard mutual interaction (the surface term) between the two subsystems (at least in the thermodynamic limit), which is only possible in finite-dimensional (i.e. non-fully connected) models. In fully connected spin systems, the surface and volume terms are of the same order, thus the former cannot be neglected. In order to avoid this, the Guerra's generalized partition function \eqref{eq:9} works by interpolating between the original model and two {\it independent} systems with sizes $N_1$ and $N_2$, such that $N_1+N_2=N$. In this way, surface terms are not present.
	\end{Remark}
}
\begin{Remark}
{The global magnetization of the composite system is a convex linear combination of the order parameters of each component, {\it i.e.}}
\begin{equation}
m_N{(\bb {\sigma}) }=\rho m_1{(\bb {\sigma}) }+ (1-\rho)m_2 {(\bb {\sigma}) },
\end{equation}
where $\rho=N_1/N$. {Further, the interpolating free energy satisfies the boundary conditions}
\begin{eqnarray}
A_N (\beta ,t=1)&=& A_N(\beta),\\
A_N (\beta ,t=0)&=& \rho A_{N_1}(\beta)+(1-\rho) A_{N_2}(\beta),
\end{eqnarray}
where of course $A_N (\beta)$ is the same as \eqref{eq:5}.
\end{Remark}

\begin{Definition}
	Given a function {$F(\bb{\sigma})$} of the spins in the {system}, its expectation value {for the interpolating system} \eqref{eq:2.1.1} is defined as
	{\begin{equation}\label{eq:2.1.2}
	\omega_t[F(\bb {\sigma})]=\frac{\sum_{\bb{\sigma}}F(\bb {\sigma})B_N(\bb\sigma)}{Z_N(\beta,t)}.
	\end{equation}}
\end{Definition}

Having introduced the central quantities, we are now in position to state the following

\begin{Theorem}
The thermodynamic limit of the model \eqref{eq:2} for even $p$ exists and it is given by
\begin{equation}
A(\beta)=\lim_{N\to \infty}A_N (\beta)=\underset{N}{\text{inf}}\, A_N(\beta),
\end{equation}

where $\text{inf}$ stands for the infimum.
\end{Theorem}
\begin{proof}
	
{
First of all, we show that, for fixed $\beta\in \mathbb R^+$, the intensive pressure $A_N(\beta)$ is bounded from below for each $N\in \mathbb N$. To see this, we write down its explicit expression:
$$
A_N(\beta)=\frac1N\log \sum _{\bb \sigma} \exp (\beta N m_N (\bb\sigma)^p).
$$
Now, since $m_N (\bb \sigma)\in [-1,1]$ for each $\bb\sigma\in \Omega_N$, we have in particular $m_N(\bb\sigma)^p\ge -1$ for each $p\in \mathbb N$. Thus, we can bound the intensive pressure as
$$
A_N (\beta)\ge \frac1N \log \sum_{\bb\sigma}\exp (-\beta N)= \frac1N \log 2^N \exp(-\beta N)=\log 2-\beta.
$$
Thus, since the sequence $\{A_N(\beta)\}_{N\in \mathbb N}$ is bounded from below for each $\beta \in \mathbb R_+$, we have $\lim_{N\to\infty} A_N (\beta)$ is finite. Next, we consider the intensive pressure associated to the interpolating partition function \eqref{eq:2.1.1}. Its $t$-derivative is}
$$
\frac{\partial}{\partial t} A_N(\beta, t)=\beta \omega_t \left(m_N{(\bb {\sigma}) }^p-\rho m_1 {(\bb {\sigma}) }^p-(1-\rho)m_2 {(\bb {\sigma}) }^p\right).
$$
Provided that $p$ is even, the mapping $x \to x^p$ is convex, meaning that

$$
m_N{(\bb {\sigma}) }^p=(\rho m_1{(\bb {\sigma}) }+ (1-\rho)m_2 {(\bb {\sigma}) })^p\le \rho m_1 {(\bb {\sigma}) }^p+ (1-\rho)m_2 {(\bb {\sigma}) }^p,
$$
$\forall \rho \in [0,1]$ and for all values of $m_1{(\bb {\sigma}) }$ and $m_2 {(\bb {\sigma}) }$. This directly implies that
$$
\frac{\partial A_N(\beta,t)}{\partial t} \le 0,
$$
thus the interpolated intensive pressure is a decreasing function w.r.t. the interpolating parameter $t$. As a straightforward consequence, we have

$$
N A_N (\beta)\le N_1 A_{N_1}(\beta)+ N_2 A_{N_2}(\beta).
$$

In other words, the sequence $\{N A_N(\beta)\}_N$ is sub-additive, and by virtue of the Fekete's lemma, we easily get
$$
\lim _{N\to \infty} A_N(\beta)=\underset{N}{\text{inf}} \, A_N (\beta)\equiv A(\beta){ \ge \log 2-\beta.}
$$
{This proves our statement.}
\end{proof}

\subsection{Solution via Guerra's interpolating scheme}

The solution of the model can be equivalently carried out via the interpolating techniques \cite{} originally developed in the context of spin-glass as alternative to the replica trick route. The key idea of the method is to introduce a generalized partition function interpolating between the original model and an effective field scenario (i.e. this limit is a 1-body model). Then, we introduce

\begin{Definition}
Let $t\in [0,1]$ an interpolating parameter. Then, the Guerra's generalized partition function is defined as
\begin{equation}
\label{eq:9}
Z_N(\beta,t)=\sum_{\bb {\sigma}}\exp \Big(t \beta N m_N{(\bb {\sigma})}^p+N(1-t)\psi m_N{(\bb {\sigma})}\Big),
\end{equation}
where $\psi$ is a constant to be set {\it a posteriori}. 
{The Boltzmann factor associated to this partition function is
$$
B_N (\bb\sigma,t)=\exp \Big(t \beta N m_N{(\bb {\sigma})}^p+N(1-t)\psi m_N{(\bb {\sigma})}\Big).
$$
}
The associated generalized intensive statistical pressure is therefore
\begin{equation} 
\label{eq:10}
A_N(\beta,t)=\frac1N\log Z_N(\beta, t).
\end{equation}
\end{Definition}
The key idea of the interpolating program is resumed in the
\begin{Proposition}
The intensive free energy of the model in the thermodynamic limit is obtained by means of the sum rule
\begin{equation}
\label{eq:11}
A(\beta)\equiv A(\beta,t=1)= \lim_{N\to \infty}\Big(A_N (\beta,t=0)+\int_0 ^1 dt' \partial_{t'}A_N(\beta,t')\Big).
\end{equation}
\end{Proposition}
\begin{proof}
The proof is a straightforward application of the fundamental theorem of integral calculus.
\end{proof}

\begin{Definition}
	Given a function {$F(\bb{\sigma})$} of the spins in the {system}, its expectation value w.r.t. to the partition function \eqref{eq:9} is defined as
{
\begin{equation}\label{eq:12}
\omega_t[F(\bb {\sigma})]=\frac{\sum_{\bb{\sigma}}F(\bb {\sigma})B_N (\bb\sigma,t)}{Z_N(\beta,t)}.
\end{equation}
}
\end{Definition}

Let us call again $\bar m$ the thermodynamic limit of the value of the global magnetization at the equilibrium, {\it i.e.}
\begin{equation}
\label{eq:2.21}
\bar m= \lim_{N\to \infty}\omega_t [m_N {(\bb {\sigma})}]. 
\end{equation}
{
\begin{Remark}
The equilibrium value $\bar m$ of the global magnetization in the thermodynamic limit can can depend on the interpolating parameter $t$. However, it is possible to show that in the thermodynamic limit, the free parameter $\psi$ can be conveniently chosen in order for $\bar m$ to be independent on $t$ almost everywhere (see next Proposition). To do this, we shall again assume the self-averaging property of the order parameter
\begin{equation}
\label{eq:14}
\lim_{N\to \infty} P(m_N(\bb {\sigma}))= \delta (m-\bar m).
\end{equation}
 $\beta$ almost everywhere and for all $t\in [0,1]$. This again implies that the variance of the global magnetization vanishes in thermodynamic limit, {\it i.e.} $\lim_{N\to \infty} (\omega_t [m_N^2]-\omega_t [m_N]^2)=0$. Further, we make the assumption that the variance of the order parameter scales as $N^{-1}$, {\it i.e.}
	\begin{equation}
		\label{eq:2.22}
	\lim_{N\to \infty}\big\vert N( \omega_t [m_N ^2]- \omega_t[m_N]^2)\big\vert <+\infty, \quad\text{almost everywhere,}
	\end{equation}
	or, in other words, fluctuations around the thermodynamic value of the order parameter scales as $1/\sqrt N$. In the physics jargon, this is equivalent to require that the magnetic susceptibility diverges only at the critical point, which is reasonable for ferromagnetic systems. The assumption \eqref{eq:2.22} will be very useful in Proposition \ref{prop:2}, since we can introduce a function $\Delta (\bb \sigma)$ accounting for fluctuations around the thermodynamic value of the global magnetization:
	\begin{equation}
	\label{eq:myexp}
	 \omega_t [(m_N- \omega_t[m_N])^2]=\mathcal O \left(\frac1N\right)\Leftrightarrow m_N(\bb\sigma )= \omega_ t [m_N (\bb \sigma)]+\frac{\Delta (\bb\sigma)}{\sqrt N},
	\end{equation}
	for sufficiently large $N$ and $\beta$ almost everywhere. The function $\Delta(\bb\sigma)$ has zero mean ($\omega_t[\Delta (\bb\sigma)]=0$, which trivially follows from taking the expectation value of both the terms in the expansion) and finite variance $\omega_t[\Delta (\bb\sigma)^2]$, even in the thermodynamic limit.
	In other words, by restricting ourselves to the pure state with positive magnetization (without loss of generality), we can express the global magnetization of a general configuration $\bb \sigma$ in terms of fluctuations (the second term on the right side in \eqref{eq:myexp}) around the expectation value $\omega_t [m_N]$. The two sides of \eqref{eq:myexp} are compatible, since with the expansion \eqref{eq:myexp} it is easy to show that
	$$
	\omega_t [(m_N- \omega_t[m_N])^2]=\frac{\omega_t [\Delta (\bb\sigma)^2]}{N},
	$$
	for sufficiently large $N$, thus satisfying the assumption \eqref{eq:2.22}.
\end{Remark}
Now, we have all of the ingredients needed to prove the following}

\begin{Proposition}\label{prop:2}
It is possible to suitably choose the parameter $\psi\in \mathbb R$ such that the expectation value in the thermodynamic limit of the global magnetization is independent on the interpolating parameter almost everywhere, {\it i.e.}
\begin{equation}
\label{eq:13}
\frac{d\bar m}{d t}=0 \quad {a.e.}
\end{equation}
\end{Proposition}

\begin{proof}
{First of all, we assume the decomposition \eqref{eq:myexp}, which we report here for the sake of completeness:}
\begin{equation}
	\label{eq:2.25}
	m_N{(\bb\sigma)}= \omega_t [m_N{(\bb\sigma)}]+\frac{\Delta{(\bb\sigma)}}{\sqrt N}.
\end{equation}
The $t$-derivative of the expectation value of the global magnetization is clearly
\begin{equation}
\label{eq:2.26}
\begin{split}
\frac{d}{d t} \omega_t [m_N {(\bb\sigma)}]&=\beta N \left(\omega_t [m_N{(\bb\sigma)}^{p+1}]-\omega_t[m_N{(\bb\sigma)}]\omega_t [m_N{(\bb\sigma)}^{p}]\right)\\
& -\psi N \left(	\omega_t [m_N{(\bb\sigma)}^2]-	\omega_t [m_N{(\bb\sigma)}]^2	\right).
\end{split}
\end{equation}
Adopting the expression \eqref{eq:2.25}, we can now compute this quantity term by term at the non-trivial order in the $1/N$ expansion. In particular
$$
	\omega_t [m_N{(\bb\sigma)}^{p+1}]=\omega_t \left[\Big(\omega_t [m_N{(\bb\sigma)}]+\frac{\Delta{(\bb\sigma)}}{\sqrt N}\Big)^{p+1}\right]=\sum_{k=0}^{p+1}\binom{p+1}{k}\omega_t [m_N{(\bb\sigma)}]^{p+1-k}\omega_t \Big[\Big(\frac{\Delta{(\bb\sigma)}}{\sqrt N}\Big)^k\Big].
$$
We are interested in considering only the contributions up to the $1/N$ order (the subleading terms will vanish in the thermodynamic limit). Then
\begin{equation}
\begin{split}
\label{eq:intermediate}
\omega_t [m_N{(\bb\sigma)}^{p+1}]&=\omega_t [m_N{(\bb\sigma)}]^{p+1}+\frac{(p+1)}{\sqrt N}\omega_t [m_N{(\bb\sigma)}]^p \omega_t [{\Delta {(\bb\sigma)}}]\\&+
\frac{p(p+1)}{2N}\omega_t [m_N{(\bb\sigma)}]^{p-1} \omega_t \Big[({\Delta {(\bb\sigma)}})^2\Big]+\mathcal R _1(m_N{(\bb\sigma)}),
\end{split}
\end{equation}
where $\mathcal R_1(m_N{(\bb\sigma)})$ accounts for the rest of the expansion, {\it i.e.}
\begin{equation}
\mathcal R_1 (m_N{(\bb\sigma)})=\sum_{k=3}^{p+1}\binom{p+1}{k}\omega_t [m_N{(\bb\sigma)}]^{p+1-k}\omega_t \Big[\Big(\frac{\Delta{(\bb\sigma)}}{\sqrt N}\Big)^k\Big].
\end{equation}
{The leading contribution in $\mathcal R_1$ is of order $N^{-3/2}$, thus the whole quantity will not contribute to \eqref{eq:2.26} in the limit $N\to\infty$}. Since the fluctuations have zero mean, the second term in \eqref{eq:intermediate} identically vanishes, leaving us only with
\begin{equation}
\label{eq:2.28}
\omega_t [m_N{(\bb\sigma)}^{p+1}]=\omega_t [m_N{(\bb\sigma)}]^{p+1}+
\frac{p(p+1)}{2N}\omega_t [m_N{(\bb\sigma)}]^{p-1} \omega_t [{\Delta {(\bb\sigma)}}^2]+\mathcal R_1 (m_N{(\bb\sigma)}).
\end{equation}
In a similar fashion, it is easy to prove that
\begin{equation}
\label{eq:2.29}
\omega_t [m_N{(\bb\sigma)}]\omega_t [m_N{(\bb\sigma)}^{p}]=\omega_t [m_N{(\bb\sigma)}]^{p+1}+
\frac{p(p-1)}{2N}\omega_t [m_N{(\bb\sigma)}]^{p-1} \omega_t [{\Delta {(\bb\sigma)}}^2]+\mathcal R_2 (m_N{(\bb\sigma)}),
\end{equation}
where also in this case $\mathcal R_2 (m_N{(\bb\sigma)})$ accounts for the subleading corrections scaling {at least as} $N^{-3/2}$ for large $N$. Finally, it is clear that
\begin{equation}
\label{eq:2.30}
\omega_t [m_N{(\bb\sigma)}^2]-	\omega_t [m_N{(\bb\sigma)}]^2= \frac1N \omega_t[\Delta{(\bb\sigma)}^2].
\end{equation}
Putting \eqref{eq:2.28}, \eqref{eq:2.29} and \eqref{eq:2.30} in \eqref{eq:2.26}, we easily get
\begin{equation}
\label{eq:2.32}
\begin{split}
\frac{d}{d t} \omega_t [m_N {(\bb\sigma)}]=\beta N &\Big(
\omega_t [m_N{(\bb\sigma)}]^{p+1}+
\frac{p(p+1)}{2N}\omega_t [m_N{(\bb\sigma)}]^{p-1} \omega_t [{\Delta {(\bb\sigma)}}^2]+\mathcal R_1 (m_N{(\bb\sigma)})
\\&-\omega_t [m_N{(\bb\sigma)}]^{p+1}-
\frac{p(p-1)}{2N}\omega_t [m_N{(\bb\sigma)}]^{p-1} \omega_t [{\Delta {(\bb\sigma)}}^2]-\mathcal R_2 (m_N{(\bb\sigma)})\Big)\\&
-\psi \omega_t [{\Delta {(\bb\sigma)}}^2]=
(\beta p \omega_t [m_N{(\bb\sigma)}]^{p-1}-\psi)\, \omega_t [\Delta{(\bb\sigma)}^2]+N \mathcal Q(m_N{(\bb\sigma)}),
\end{split}
\end{equation}
where {we defined} $\mathcal Q (m_N{(\bb\sigma)})=\mathcal R_1 (m_N{(\bb\sigma)})-\mathcal R_2 (m_N{(\bb\sigma)})$ {whose leading contribution} scales itself as $N^{-3/2}$ in the large $N$ limit (thus, $N \mathcal Q (m_N{(\bb\sigma)}$ scales as $N^{-1/2}$){, and therefore it is negligible in the $N\to\infty$ limit}). The r.h.s. of the last line in \eqref{eq:2.32} is well-defined in the thermodynamic limit, thus - calling $\bar m = \lim_{N\to \infty} \omega_t[m_N [\bb \sigma]]$, we have (recall that the variable $\Delta$ has finite variance in the $N\to\infty$ limit)
\begin{equation}
\lim_{N\to\infty}\frac{d}{d t} \omega_t [m_N {(\bb\sigma)}]=(\beta p \bar m^{p-1}-\psi)\, \lim_{N\to\infty}\omega_t [\Delta{(\bb\sigma)}^2].
\end{equation}
{Here, we used the fact that $N \mathcal Q\to 0$ in the limit $N\to\infty$, since the leading contribution of the quantity $\mathcal Q$ is of order $N^{-3/2}$.}
This means that the sequence $\{\frac{d\omega_ t[m_N{(\bb\sigma)}]}{dt}\}_{N}$ converges almost everywhere to the r.h.s. of the previous equation, so by virtue of Egorov's theorem \cite{zbMATH02629860}, it is almost uniformly convergent. As a consequence, the relation
$$
\lim_{N\to\infty}\frac{d}{d t} \omega_t [m_N{(\bb\sigma)}]=\frac{d}{d t}\lim_{N\to\infty} \omega_t [m_N {(\bb\sigma)}]
$$
holds almost everywhere, which means
\begin{equation}
\frac{d\bar m}{d t}=(\beta p \bar m^{p-1}-\psi)\, \lim_{N\to\infty}\omega_t [\Delta{(\bb\sigma)}^2]
\end{equation}
Then, it is {simple} to note that the requirement $\frac{d\bar m}{dt}=0$ can be consistently fulfilled almost everywhere by choosing
$$
\psi=\beta p \bar m^{p-1},
$$
which proves our assertion.
\end{proof}

\begin{Proposition}
The expectation value of the $p$-th power of the global magnetization can be written in terms of the centered momenta with degree lower or equal to $p$.
\end{Proposition}

\begin{proof}
The proof is a straightforward application of binomial theorem. Indeed
$$
\omega_t [m_N {(\bb\sigma)}^p]=\omega_t [(\bar m+ m_N {(\bb\sigma)}-\bar m)^p]=\sum_{k=0}^p \binom{p}{k} \omega_t [( m_N{(\bb\sigma)}-\bar m)^k] \bar m^{p-k}.
$$
\end{proof}

\begin{Remark}\label{rem:6}
We stress that we can also extract the lower two terms from the sum, in order to get
$$
\omega_t [m_N {(\bb\sigma)}^p]=\bar m^p+p  \, \bar m^{p-1} \, \omega_t [( m_N {(\bb\sigma)}-\bar m)]+\sum_{k=2}^p \binom{p}{k} \omega_t [( m_N {(\bb\sigma)}-\bar m)^k] \bar m^{p-k},
$$
or in other words
\begin{equation}
\label{eq:15}
\omega_t [m_N{(\bb\sigma)}^p]-p \bar m^{p-1}\omega_t [ m_N {(\bb\sigma)}]=(1-p)\bar m^p+\sum_{k=2}^p \binom{p}{k} \omega_t [( m_N {(\bb\sigma)}-\bar m)^k] \bar m^{p-k}.
\end{equation}
\end{Remark}
With all of these ingredients in our hand, we can prove the following
\begin{Theorem}
The thermodynamic limit of the intensive statistical pressure for the model \eqref{eq:2} is given by
\begin{equation}
\label{eq:16}
A(\beta)=\beta(1-p)\bar m^p +\log 2+\log \cosh (\beta p \bar m^{p-1}).
\end{equation}
\end{Theorem}

\begin{proof}
First of all, we compute the derivative of the statistical pressure. In this case, we have
$$
\frac{\partial A_N (\beta ,t)}{\partial t}= \beta \omega_t [m_N {(\bb\sigma)}^p]-\psi \omega_t [m_N {(\bb\sigma)}].
$$
We notice that, recalling our choice $\psi= \beta p \bar m^{p-1}$, we have
$$
\frac{\partial A_N (\beta ,t)}{\partial t}=\beta\left( \omega_t [m_N {(\bb\sigma)}^p]-p \bar m ^{p-1} \omega_t [m_N {(\bb\sigma)}]\right),
$$
so that we can apply the Rem. \ref{rem:6}. Indeed, using the relation \eqref{eq:15}, we easily get
$$
\frac{\partial A_N (\beta ,t)}{\partial t}=\beta (1-p)\bar m^p+\beta \sum_{k=2}^p \binom{p}{k} \omega_t [( m_N {(\bb\sigma)}-\bar m)^k] \bar m^{p-k}.
$$
We stress that, by assuming the self-averaging property \eqref{eq:14} of the order parameter, the sum in r.h.s. clearly vanishes in the thermodynamic limit (since centered momenta will disappear in the $N\to\infty$ limit), therefore leaving us only with
\begin{equation}
\label{eq:17}
\lim_{N\to \infty}\frac{\partial A_N (\beta ,t)}{\partial t}= \beta (1-p )\bar m^p,
\end{equation}
which is independent on $t$ (due to Prop. \ref{prop:2}). Thus, its $t$-integration is trivial. On the other side, the $t=0$ is easy to handle with, since it is a 1-body computation. Indeed
$$
A_N(\beta ,t=0)=\frac1N\log \sum_{\bb {\sigma}}\exp\big(N \psi m_N {(\bb\sigma)} \big)=\frac1N\log \sum_{\sigma_1=\pm1}\dots \sum_{\sigma _N=\pm1}\exp\big(\psi \sum_i \sigma_i\big),
$$
which leads to
\begin{equation}
\label{eq:18}
A_N(\beta ,t=0)=\frac1N \log 2^N \cosh \psi= \log 2+\log \cosh (\beta p \bar m^{p-1}),
\end{equation}
where we recalled our choice $\psi=\beta p \bar m^{p-1}$. Now, using the sum rule \eqref{eq:11}, we obtain the thermodynamic limit of the intensive pressure
\begin{equation}
\label{eq:19}
A(\beta)=\beta(1-p)\bar m^p +\log 2+\log \cosh (\beta p \bar m^{p-1}),
\end{equation}
as claimed.
\end{proof}

\begin{Corollary}
The self-consistency equation for the global magnetization is
\begin{equation}
\label{eq:20}
\bar m = \tanh (\beta p \bar m^{p-1}).
\end{equation}
\end{Corollary}

\begin{proof}
The derivation of the self-consistency equation immediately follows from the extremality condition for the statistical pressure $\partial _{\bar m}A(\beta)=0$.
\end{proof}

This is in agreement with the results coming from purely statistical mechanics arguments \eqref{eq:8}.

\section{Few words on Burgers hierarchy}\label{sec:3}

The Burgers equation \cite{burgers1,BURGERS1948171,whitham} is one of the most studied nonlinear evolutive equations, and it can be written in the form
\begin{equation}
\label{eq:3.1}
u_t+ 2u u_x +\alpha u_{xx}=0,
\end{equation}
where $\alpha $ is the {\it viscosity} parameter {and $(\dots)_t\equiv \partial _t (\dots)$, $(\dots)_x\equiv \partial _x (\dots)$, $(\dots)_{xx}\equiv \partial^2 _x (\dots)$ and so on}. This equation is known to emerge as a $1+1$-dimensional reduction of Navier-Stokes equations for an incompressible fluid in absence of pressure gradient \cite{whitham,WAZWAZ2}. The most important peculiarity of this equation is that it can be linearized in the heat equation via Cole-Hopf transform \cite{cole,hopf}. Furthermore, it is one of the simplest models describing the development and propagation of shock waves. A related and well-studied equation is the Sharma-Tasso-Olver (STO) equation \cite{olver,Sharma,tasso1976cole}, which can be written as
\begin{equation}
\label{eq:3.2}
u_t + 3u^2 u_x +3 \alpha u_x ^2 +3\alpha u u_{xx}+\alpha^2 u_{xxx}=0.
\end{equation} 
This equation is known to be integrable (in particular, it presents infinitely many symmetries, a bi-Hamiltonian structure, solitary wave solution and an infinite number of conservation laws \cite{tasso1976cole,Sharma,olver,WAZWAZ20071205}).
\par
These two equations are the lowest elements of the so-called Burgers hierarchy, which can be presented in the form
\begin{equation}
\label{eq:3.3}
	\frac{\partial u(t,x)}{\partial t}+\frac{\partial}{\partial x}\left( \alpha\frac{\partial}{\partial x}+u(t,x)\right)^{n}u(t,x)=0,\quad n=1,2,\dots.
\end{equation}
The next two higher equations in the hierarchy are respectively
\begin{eqnarray}
u_t +4 u^3 u_x +12 \alpha u u_x ^2+6\alpha u^2 u_{xx}+10 \alpha^2 u_x u_{xx}+4\alpha^2 u u_{xxx}+\alpha^3 u_{xxxx}&=&0,\\
u_t+5u^4 u_x+30\alpha u^2 u_x^2 +15 \alpha^2 u_x^3+10\alpha u^3 u_{xx}+50\alpha^2 u u_x u_{xx}+10 \alpha^3 u_{xx}^2\notag \\+10 \alpha^2 u^2 u_{xxx}+15 \alpha^3 u_x u_{xxx}+5\alpha^3 u u_{xxxx}+\alpha^4 u_{xxxxx}&=&0.
\end{eqnarray}
It is clear that the complexity of the elements in the Burgers hierarchy dramatically increases with the index $n$. However, all of these equations share the same property of Burgers equation, see also \cite{olver}.

\begin{Remark}
Notice that the Burgers hierarchy is also commonly written in the form
\begin{equation}
	\frac{\partial u(t,x)}{\partial t}+\delta\frac{\partial}{\partial x}\left( \frac{\partial}{\partial x}+u(t,x)\right)^{n}u(t,x)=0.
\end{equation}
This is related to the form given in \eqref{eq:3.3} is given by applying on the former the transformation $x\to \delta x$ with $\delta =\frac1\alpha$.
\end{Remark}

\begin{Lemma}\label{lem:1}
The following identity holds:
\begin{equation}
\left(\frac{\partial}{\partial x}+\frac{\Psi_x}{\Psi}\right)^n \frac{\Psi_x}{\Psi}=\frac{\Psi_{n+1,x}}{\Psi},
\end{equation}
where $\Psi_{n+1,x}=\partial_x ^{n+1}\Psi$.
\end{Lemma}
\begin{proof}
The proof of this Lemma works in the same way of Lemma 1 in \cite{KUDRYASHOV20091293}. However, in order to make this Section self-contained, we reported here for the sake of completeness. The proof works by induction, so let us prove it for $n=1$ first:

\begin{equation}
\left(\frac{\partial}{\partial x}+\frac{\Psi_x}{\Psi}\right)\frac{\Psi_x}{\Psi}=\frac{\Psi_{xx}}{\Psi}-\frac{\Psi_x^2}{\Psi^2}+\frac{\Psi_x^2}{\Psi^2}=\frac{\Psi_{xx}}{\Psi}.
\end{equation}
Assuming now that the identity holds for $n$, we will prove it for $n+1$. Indeed:
\begin{equation}
\left(\frac{\partial}{\partial x}+\frac{\Psi_x}{\Psi}\right)^{n+1} \frac{\Psi_x}{\Psi}=\left(\frac{\partial}{\partial x}+\frac{\Psi_x}{\Psi}\right)\left(\frac{\partial}{\partial x}+\frac{\Psi_x}{\Psi}\right)^n \frac{\Psi_x}{\Psi}.
\end{equation}
Using the thesis for $n$, the last member of the equation is
\begin{equation}
\begin{split}
\left(\frac{\partial}{\partial x}+\frac{\Psi_x}{\Psi}\right) \frac{\Psi_{n+1,x}}{\Psi}&=\left(\frac{\partial}{\partial x}+\frac{\Psi_x}{\Psi}\right) \frac{\Psi_{n+1,x}}{\Psi}=\\&=\left(\frac{\Psi_{n+2,x}}{\Psi}-\frac{\Psi_x \Psi_{n+1,x}}{\Psi^2}+\frac{\Psi_x \Psi_{n+1,x}}{\Psi^2}\right)=\frac{\Psi_{n+2,x}}{\Psi},
\end{split}
\end{equation}
which proves our assertion.
\end{proof}

\begin{Lemma}\label{lem:2}
The following identity holds:
\begin{equation}
\frac{\partial}{\partial t}\frac{\Psi_x}{\Psi}=\frac{\partial}{\partial x}\frac{\Psi_t}{\Psi}.
\end{equation}

\end{Lemma}
\begin{proof}
The proof works by straightforward computation. Indeed
\begin{equation}
\frac{\partial}{\partial t}\frac{\Psi_x}{\Psi}=\frac{\Psi_{x,t}}{\Psi}-\frac{\Psi_x \Psi_t}{\Psi^2},
\end{equation}
while
\begin{eqnarray}
\frac{\partial}{\partial x}\frac{\Psi_t}{\Psi}=\frac{\Psi_{t,x}}{\Psi}-\frac{\Psi_x \Psi_t}{\Psi^2}.
\end{eqnarray}
Assuming that the function $\Psi$ is analytic in $x$ and $t$, then $\Psi_{x,t}= \Psi_{t,x}$, leading to
\begin{equation}
\frac{\partial}{\partial t}\frac{\Psi_x}{\Psi}=\frac{\partial}{\partial x}\frac{\Psi_t}{\Psi}.
\end{equation}
\end{proof}
We are now in position to state the following
\begin{Theorem}
Each element of the Burgers hierarchy can be linearized via Cole-Hopf transform into linear equations.
\end{Theorem}

\begin{proof}
	First, we perform the Cole-Hopf transform
	\begin{equation}
	u(t,x)=\alpha\frac{\Psi_x}{\Psi}\equiv \alpha (\log \Psi)_x.
	\end{equation}
By plugging it into the Burgers hierarchy \eqref{eq:3.3} we have
\begin{equation}
\frac{\partial}{\partial t}\alpha\frac{\Psi_x}{\Psi}
+\frac{\partial }{\partial x} \left(\alpha\frac{\partial}{\partial x}+\alpha\frac{\Psi_x}{\Psi}\right)^n \alpha\frac{\Psi_x}{\Psi}=\alpha \left(\frac{\partial}{\partial t}\frac{\Psi_x}{\Psi}
+\frac{\partial }{\partial x} \left(\alpha\frac{\partial}{\partial x}+\alpha\frac{\Psi_x}{\Psi}\right)^n \frac{\Psi_x}{\Psi}\right).
\end{equation}
Now, using Lemmas \ref{lem:1} and \ref{lem:2}, we have
\begin{equation}
\frac{\partial u(t,x)}{\partial t}+\frac{\partial}{\partial x}\left( \alpha\frac{\partial}{\partial x}+u(t,x)\right)^{n}u(t,x)=\alpha \frac{\partial }{\partial x}\left(\frac{\Psi_t }{\Psi}+\alpha^n\frac{\Psi _{n+1,x}}{\Psi}\right)=0.
\end{equation}
By setting the argument of the derivative to zero and assuming that $\Psi \neq 0$ for all $x$ and $t$, we finally get the linear equations
\begin{equation}
\label{eq:3.18}
\Psi_t +\alpha^n \Psi_{n+1,x}=0.
\end{equation}
\end{proof}

\section{Guerra's mechanical scheme and relation with the Burgers hierarchy}\label{sec:4}
{This final Section is devoted to prove the connection between the Guerra's interpolated partition function of the $p$-spin ferromagnets and the Burgers hierarchy. Before proceeding, it is worth to mention that relations between PDEs and statistical models have been extensively analyzed in the literature, for instance by Ellis and Newman \cite{newman} and Bogolyubov and co-workers \cite{Bogolyubov,Brankov}. More recently, PDE methods have gained a consistent role in the analysis of statistical spin models, especially for disordered systems, see for example \cite{choquard,Genovese,barrapde1,barrapde2} and references therein.}\par\medskip
Having introduced both the players in the duality, we are now in position to prove it. We will start by defining the generalized quantities in the Guerra's interpolation scheme.

\begin{Definition}
The Guerra's generalized partition function is defined as
\begin{equation}
\label{eq:21}
Z_N(t,x)=\sum_{\bb {\sigma}}\exp\Big(-t Nm_N {(\bb\sigma)}^p+N x m_N{(\bb\sigma)}\Big),
\end{equation}
with associated {Boltzmann factor
$$
B_N(t,x)=\exp\Big(-t Nm_N {(\bb\sigma)}^p+N x m_N{(\bb\sigma)}\Big).
$$
The} intensive statistical pressure {of the model is}
\begin{equation}
\label{eq:22}
A_N (t,x)=\frac1N \log Z_N(t,x).
\end{equation}
\end{Definition}

\begin{Remark}
We interpret the variable $t$ and $x$ respectively as temporal and spatial coordinates in a 1+1-dimensional space. This interpretation will be clear in a moment.
\end{Remark}

{
\begin{Remark}
We stress that the original $p$-spin model \eqref{eq:2} (without external fields)  is recovered with the choice $t=-\beta$ and $x=0$. The inclusion of a (uniform) magnetic field is reproduced by setting $x=h\neq 0$, so that the present framework will still work.
\end{Remark}
}

\begin{Definition}
Given a function {$F(\bb {\sigma})$} of the spins, its expectation value {for the interpolating system} \eqref{eq:21} is defined as
{
\begin{equation}
\label{eq:23}
\omega_{t,x}[F(\bb {\sigma})]=\frac{\sum_{\bb{\sigma}}F(\bb {\sigma})B_N (t,x)}{Z_N (t,x)}.
\end{equation}
}
\end{Definition}

We assume that the function $A_N(t,x)$ is a differentiable function w.r.t. the space-time coordinates. We can therefore compute its derivatives. In particular, we have
\begin{eqnarray}
\label{eq:24}
\frac{\partial A_N (t,x)}{\partial t}&=&-\omega_{t,x} [m_N{(\bb\sigma)}^p],\label{eq:24.1}\\
\frac{\partial A_N (t,x)}{\partial x}&=&\omega_{t,x} [m_N{(\bb\sigma)}].\label{eq:24.2}
\end{eqnarray}
In order to find differential equations, we also need higher spatial derivative of the generalized statistical pressure (or equivalently, derivatives of the magnetization expectation value by virtue of \eqref{eq:24.2}). It is trivial to note that differentiating the expectation value of the global magnetization would generate expectation values of polynomial in the magnetization itself. This is due to the fact that the $x$-derivative should increase of a unity the power. This is clear, for example, by considering the first derivative of the expectation value of the magnetization. Indeed, it is clear that
$$
\frac{\partial \omega_{t,x}[m_N {(\bb\sigma)}]}{\partial x}= N(\omega_{t,x}[m_N {(\bb\sigma)}^2]-\omega_{t,x}[m_N {(\bb\sigma)}]^2).
$$
A similar relation holds for the expectation value of a generic power of the magnetization:
$$
\frac{\partial \omega_{t,x}[m_N {(\bb\sigma)}^q]}{\partial x}= N(\omega_{t,x}[m_N {(\bb\sigma)}^{q+1}]-\omega_{t,x}[m_N {(\bb\sigma)}^q]\omega_{t,x}[m_N {(\bb\sigma)}]).
$$
If we call $u_q(t,x) =\omega_{t,x}[m_N{(\bb\sigma)}^q]$, then we have the following

\begin{Proposition}
The structure of the expectation values of the powers of magnetization is resumed in the following relation
\begin{equation}
\label{eq:25}
\frac{\partial u_q(t,x)}{\partial x}=N(u_{q+1}(t,x)-u_q (t,x)u _1 (t,x)).
\end{equation}
\end{Proposition}

\begin{Proposition}
The functions $u_q(t,x)$ satisfy the recurrence relation
\begin{equation}
\label{eq:26}
u_{q+1}(t,x)=\frac1N \frac{\partial u_q (t,x)}{\partial x}+u_q(t,x)u_1(t,x).
\end{equation}
\end{Proposition}

\begin{proof}
	The proof of this Proposition follows from a trivial rearrangement of Eq. \eqref{eq:25}.
\end{proof}

{Practically}, we generate the expectation value by repeatedly applying the operator $N^{-1}\partial_x +u_1(t,x)$ on the expectation value of the magnetization $u_1 (t,x)=\omega_{t,x}[m_N {(\bb\sigma)}]$:

\begin{equation}
\label{eq:27}
u_{q+1}(t,x)=\left( \frac1N\frac{\partial}{\partial x}+u_1(t,x)\right)^q u_1 (t,x).
\end{equation}

Now, taking $q+1=p$ and recalling that $u_p(t,x)=-\partial _t A_N(t,x)$, we have

$$
\frac{\partial A_N(t,x)}{\partial t}+\left( \frac1N\frac{\partial}{\partial x}+u_1(t,x)\right)^q u_1 (t,x)=0.
$$
{Further, taking} the $x$-derivative of the entire equation, commuting the derivatives $\partial_t $ and $\partial_x$ acting on $A_N(t,x)$ and recalling that $\partial_x A_N(t,x)=u_1 (t,x)$ (which we simply call $u(t,x)$ for the sake of simplicity), we arrive to state the following

\begin{Theorem}
For each $p\ge 2$, the expectation value of the global magnetization $u(t,x)$ of the $p$-spin model described by the Guerra's generalized partition function \eqref{eq:21} satisfies the equation of the Burgers hierarchy
\begin{equation}
\label{eq:28}
\frac{\partial u(t,x)}{\partial t}+\frac{\partial}{\partial x}\left( \alpha\frac{\partial}{\partial x}+u(t,x)\right)^{p-1}u(t,x)=0,
\end{equation}
where $\alpha=N^{-1}$ is the viscosity parameter.
\end{Theorem}

\begin{Remark}
An alternative route for the proof of the duality is to start with the partition function $Z_N(t,x)$ and compute the space-time derivatives. It is easy to see that
$$
\frac{\partial Z_N(t,x)}{\partial t}+\frac{1}{N^{p-1}} \frac{\partial ^p Z_N(t,x)}{\partial x^p}=0.
$$
In this setup, the relation between the expectation value of the magnetization $\omega_{t,x} (m_N)$ and the partition function $Z_N(t,x)$ is precisely the Cole-Hopf transform, then the duality is easily understood.
\end{Remark}

\subsection{The inviscid limit and gradient catastrophe}

From the point of view the duality with the Burgers hierarchy, the thermodynamic limit corresponds to the inviscid limit of the Burgers hierarchy $\alpha\to0$. Then, the whole class of non-linear equations dramatically simplifies, so that we get
$$
\frac{\partial u(t,x)}{\partial t}+\frac{\partial}{\partial x}u(t,x)^p=\frac{\partial u(t,x)}{\partial t}+ p u(t,x)^{p-1}\frac{\partial u(t,x)}{\partial x}=0.
$$
In order to solve this equation, we also need the initial profile of the solution $u_0(x)=u(t=0,x)$. Again, this is a trivial computation, since the initial profile is a 1-body problem. Indeed, we have
\begin{equation*}
\begin{split}
u_0 (x)&= \frac{\partial}{\partial x}\frac 1N \log Z_N(t=0,x)=\frac{\partial}{\partial x}\frac 1N \log\sum_{\bb \sigma} \exp(N x m_N{(\bb\sigma)})=\\&=
\frac{\partial}{\partial x}\frac 1N \log\sum_{\bb \sigma} \exp\big(x\sum_{i=1}^N\sigma _i \big)=\frac{\partial}{\partial x}\frac 1N \log 2^N \cosh ^N(x)=\tanh (x).
\end{split}
\end{equation*}
Therefore, we can state the following
\begin{Proposition}
The solution of the self-consistency equation in the thermodynamic limit of the model \eqref{eq:21} is given by the solution of the initial value problem
\begin{equation}
\label{eq:29}
\begin{cases}
\frac{\partial u(t,x)}{\partial t}+ p\, u(t,x)^{p-1}\frac{\partial u(t,x)}{\partial x}=0\\
u_0(x)=u(t=0,x)=\tanh(x)
\end{cases}.
\end{equation}
\end{Proposition}
We stress that such a (first order) system describes the motion of traveling waves in $1+1$ dimensions with effective velocity $v(t,x)=p\, u(t,x)^{p-1}$. Then, as standard in this case, we can look for solution in implicit form as $u(t,x)=u_0 (x-v(t,x)t)=\tanh(x-v(t,x)t)$. Now, recalling that the original $p$-spin model is recovered with the choice $x=0$ and $t=-\beta$ and that $u(-\beta,0)=\lim_{N\to\infty} \omega_{-\beta,0} (m_N{(\bb\sigma)})=\bar m$ in the inviscid limit (here, we dropped the dependency on $t$ and $x$, but we again assume the self-averaging property of the order parameter), we easily get the self-consistency equation
\begin{equation}
\label{eq:30}
\bar m = \tanh (\beta p\, \bar m^{p-1}),
\end{equation}
in perfect agreement with previous results \eqref{eq:8} and \eqref{eq:20}.

\begin{Remark}
For each $p$, the model equilbrium dynamics (as resumed in \eqref{eq:30}) undergoes an ergodicity breaking transition. However, this phase transition is of second order (in the standard Erhenfest classification, {\it i.e.} the second derivative of the free energy is discontinuous) only for $p=2$, while for all $p>2$ the phase transition is of first order ({\it i.e.} the first derivative of the free energy is discontinuous).
\end{Remark}

On the Burgers' side, the inviscid limit has the peculiarity of the appearance of the gradient catastrophe and the related development of shock waves. Indeed, it is easy to understand that these two phenomenons ({\it i.e.} gradient catastrophe and ergodicity breaking) are equivalent in this mapping. To understand this, let us analyze the characteristic curves of the system \eqref{eq:29}, given by the system
{
$$
\frac{dt}{1}= \frac{dx}{p u(t,x)^{p-1}},
$$
while $u$ is constant along the characteristic curves, {\it i.e.} $du/dt=0$.}
The characteristic curves are straight lines which can be parametrized in terms of a quantity $\xi$ as $x_\xi (t)=\xi + F(\xi) t$, where $\xi$ clearly is the $x$ value at $t=0$ and $F(\xi)= v(u_0 (\xi))= p \tanh (\xi)^{p-1}$. It is well known that the gradient of the solution of \eqref{eq:29} diverges as $t_c(\xi)=-1/F'(\xi)$, where the characteristic lines start to cross each other.

\begin{Remark}
By inspecting at its plot, we see that, for even $p$, the function $F'(\xi)$ is only positive, meaning that the gradient catastrophe only takes place for $t<0$. This is consistent with our approach, since $t=-\beta$ and $\beta\in \mathbb R^+$. On the other hand, rigorously the solution of \eqref{eq:29} is defined for $t>0$. This is not a problem, since the solution can be analytical continuated for negative values of $t$ just before the shock ({\it i.e.} for $t>t_c$). A sketch of the family of characteristic curves is depicted in Fig. \ref{fig:self_cs}, left panel. 
\end{Remark}
Since we ultimate want to set $x=0$ and $t=-\beta$, we should consider only the characteristics for which the gradient catastrophe holds at $x=0$. In order to ensure this, we have to choose $x_{\bar \xi }(t_c)=0$ for some $\bar \xi$, meaning that $\bar \xi=-F(\bar \xi )t_c$. By using the formula $t_c (\xi)=-1/F'(\xi)$, we have that $\bar \xi= F(\bar \xi)/F'(\bar \xi)$. In other words, with the last formula we find for the values of $\xi$ for which the gradient explodes in the position $x=0$, then with the standard gradient formula we can compute the time at which the solution develops a shock wave. Recalling that $t=-\beta =-\frac1T$, we arrive at the following

\begin{Proposition}
The critical temperature for the ergodicity breaking phase transition can be identified resolving the following system:

\begin{equation}
\label{eq:31}
\begin{cases}
\bar \xi= F(\bar \xi)/F'(\bar \xi)\\
T_c =F'(\bar \xi)
\end{cases},
\end{equation}
where $F(\xi)=p \tanh(\xi)^{p-1}$.
\end{Proposition}

The results are reported in Fig. \ref{fig:self_cs}, right panel.
\begin{figure}[t]
	\begin{minipage}{0.52\linewidth}	
		\includegraphics[width=\textwidth]{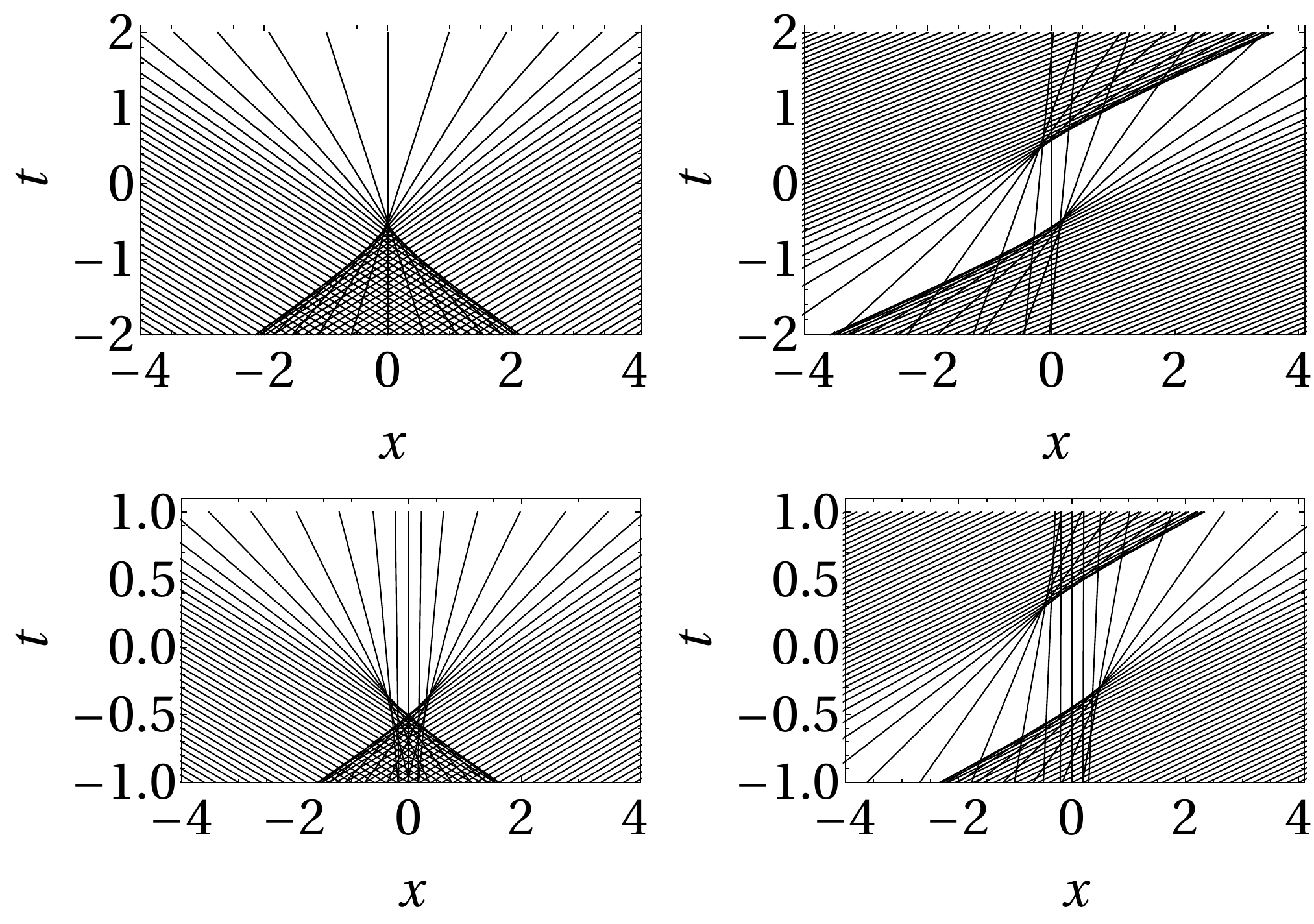}
	\end{minipage}
	\hspace{0.1cm}
	\begin{minipage}{0.45\linewidth}	
		\includegraphics[width=\textwidth]{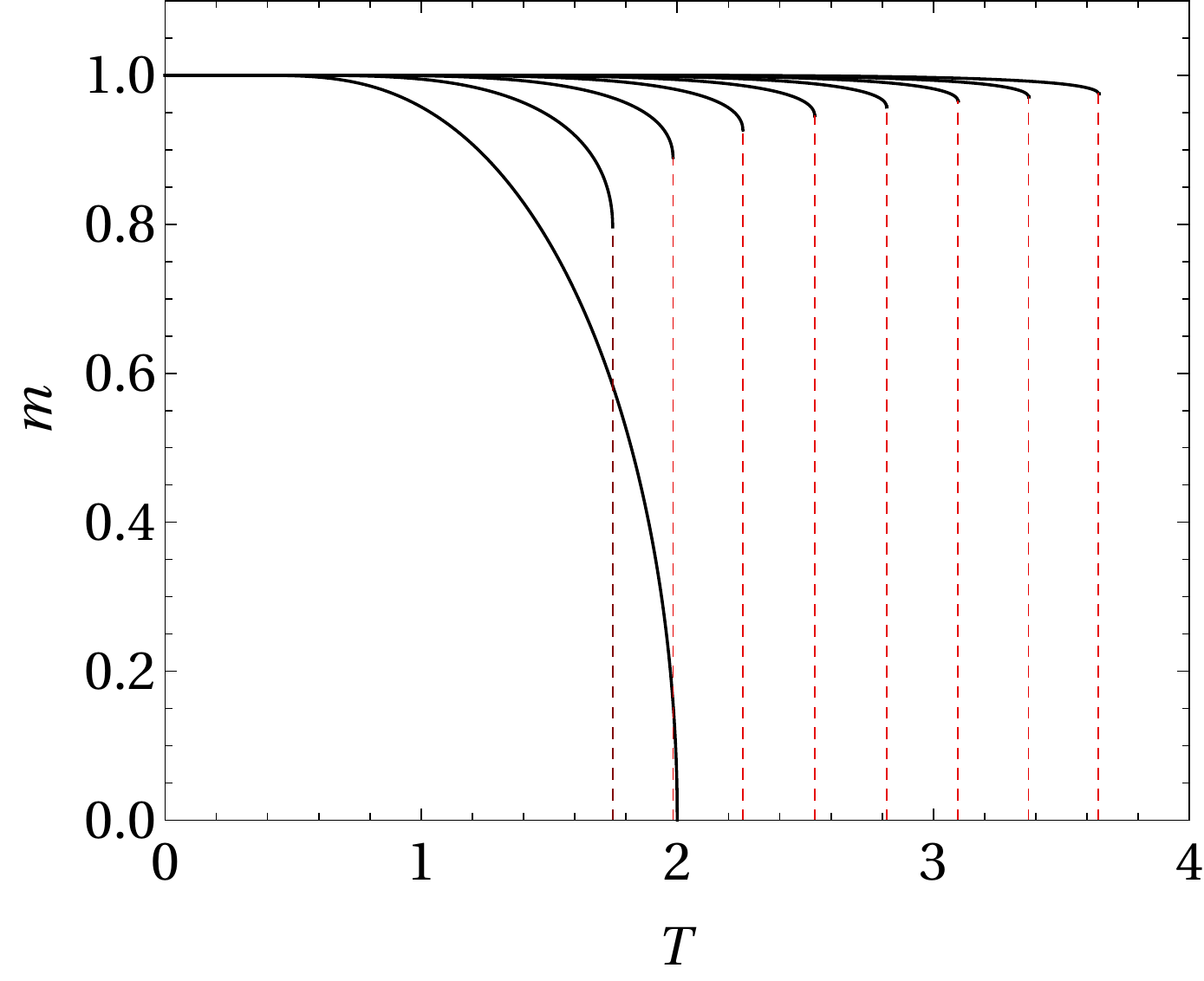}
		\vspace{-0.85cm}
	\end{minipage}
	\caption{(Left panel). Representation of characteristic curves for the system \eqref{eq:29}. In particular, we have $p=2$ (upper left plot), $p=3$ (upper right plot), $p=4$ (lower left plot) and $p=5$ (lower right plot). We see the different patterns of the characteristics curve for even $p$ (mutual crossing of characteristics only for $t<0$) and for odd $p$ (mutual crossing of characteristics both for $t<0$ and $t>0$). Of course, only the $t<0$ gradients catastrophes are relevant in this scenario. (Right panel). Numerical solution (black solid lines) of the self-consistency equation \eqref{eq:30} for $p=2,3,4,\dots,10$ (going from the left to the right). The red dashed vertical lines represent the critical temperature for ergodicity breaking phase transition as predicted by the system \eqref{eq:31}. Numerical results and theoretical predictions are in perfect agreement.}
	\label{fig:self_cs}
\end{figure}

\subsection{Burgers hierarchy solution from $p$-spin thermodynamics}

In this Section, we will take benefit of the duality in order to find explicit solutions of the Burgers hierarchy with non-vanishing viscosity $\alpha$ and initial profile $u_0(x)=\tanh x$. To this aim, we should find explicit expressions for the expectation value of the global magnetization at finite $N$. In this case, the global magnetization $m_N{(\bb\sigma)}$ can only take discrete values. To understand this, we can start with a system configuration in which $\sigma _i =1$ for all $i=1,\dots,N$, whose corresponding magnetization is trivially $m_N{(\bb\sigma)}=1$. All of the other values can be obtained by progressively flipping all of the spins until the lower bound $m_N{(\bb\sigma)}=-1$ is reached (corresponding to a situation in which $\sigma_i =-1$ for all $i=1,\dots,N$). Every time we flip a spin, there will be a net difference in the value of the magnetization whose magnitude is $2/N$. Therefore, the possible values of the magnetization are
\begin{equation}
m_N \in\Big\{ 1, \frac{N-2}{N}, \frac{N-4}{N},\dots,-\frac{N-4}{N},- \frac{N-2}{N},-1\Big\},
\end{equation}
or in compact form $ m_{N,k}=\frac1N(N-2k)$ for $k=0,\dots,N$. Each possible value of the magnetization has a degeneracy given by
$$
\text{Deg}(m_{N,k})= \binom{N}{k}.
$$
With these ingredients, we can write the generalized partition function \eqref{eq:21} at finite $N$ as
\begin{equation}
\label{eq:4.14}
Z_N(t,x)= \sum_{k=0}^N\binom{N}{k}\exp\Big(-\frac{t}{N^{p-1}}(N-2k)^p+x(N-2k)\Big).
\end{equation}
Now, using the basic relation $\omega_{t,x}[m_N {(\bb\sigma)}]=\frac1N \partial_x \log Z_N(t,x)$ and the duality result $u(t,x)=\omega_{t,x}[m_N {(\bb\sigma)}]$, we have
\begin{equation}
\label{eq:4.15}
u_{N,p}(t,x)=\frac1N \frac{\sum_{k=0}^N\binom{N}{k}(N-2k)\exp\Big(-\frac{t}{N^{p-1}}(N-2k)^p+x(N-2k)\Big)}{\sum_{k=0}^N\binom{N}{k}\exp\Big(-\frac{t}{N^{p-1}}(N-2k)^p+x(N-2k)\Big)}.
\end{equation}
In the last equation, we used the subscripts $p$ to distinguish between different elements of the Burgers hierarchy and $N$ to stress that the corresponding viscosity parameter is fixed as $\alpha=1/N$. {As is clear from \eqref{eq:4.15}, the solutions we can build by explicit use of the duality are rational functions in which both the numerator and denominator are linear combinations (with coefficients not depending on $x$ and $t$) of exponential waves of the form $e^{Ax+Bt}$. Other solutions sharing this structure can be found in \cite{KUDRYASHOV20091293}.} To conclude this analysis, we provide some specific examples of solutions of the initial value problem for fixed $N$ and $p$.\par\medskip
For $p=N=2$, the function
\begin{equation*}
u_{2,2}(t,x)=\frac{-1 +e^{4x}}{1+e^{4x}+2 e^{2t+2x}},
\end{equation*}
is solution of the initial value problem
\begin{equation*}
\begin{cases}
u_t+2 u u_x +\frac12 u_{xx}=0\\
u_0(x)=u(t=0,x)=\tanh(x)
\end{cases}.
\end{equation*}
For $p=2$, $N=3$, the function
\begin{equation*}
u_{2,3}(t,x)=\frac{-1 +e^{6x}-e^{\frac83 t+2x}+e^{\frac83t+4x}}{ 1+e^{6x}+3 e^{\frac83 t+2x}+3 e^{\frac83 t+4x}},
\end{equation*}
is solution of the initial value problem
\begin{equation*}
\begin{cases}
u_t+2 u u_x +\frac13 u_{xx}=0\\
u_0(x)=u(t=0,x)=\tanh(x)
\end{cases}.
\end{equation*}
For $p=2$, $N=4$, the function
\begin{equation*}
u_{2,4}(t,x)=\frac{-1 +e^{8x}-2e^{3t+2x}+2 e^{3t+6x} }{ 1 +e^{8x}+4 e^{3t+2x}+6 e^{4t +4x}+4 e^{3t+6x}  },
\end{equation*}
is solution of the initial value problem
\begin{equation*}
\begin{cases}
u_t+2 u u_x +\frac14 u_{xx}=0\\
u_0(x)=u(t=0,x)=\tanh(x)
\end{cases}.
\end{equation*}
We can also vary the value of the order $p$ of the interactions. Indeed, for $p=3$ and $N=2$, the function
\begin{equation*}
u_{3,2}(t,x)=\frac{  -e^{2t}+e^{2x}  }{   e^{2t}+e^{2x}   },
\end{equation*}
is solution of the initial value problem
\begin{equation*}
\begin{cases}
u_t+ 3u^2 u_x +\frac32 u_x^2+\frac 32 uu_{xx}+\frac14 u_{xxx}=0\\
u_0(x)=u(t=0,x)=\tanh(x)
\end{cases}.
\end{equation*}
Finally, for $p=4$ and $N=2$, the function
\begin{equation*}
u_{4,2}(t,x)=\frac{-1 +e^{4x} }{ 1+e^{4x}+2e^{2t+2x}},
\end{equation*}
is solution of the initial value problem
\begin{equation*}
\begin{cases}
u_t+ 	4u^3 u_x +6 u u_x^2 +3 u^2 u_{xx}+\frac52 u_x u_{xx}+u u_{xxx}+\frac18 u_{xxxx}=0\\
u_0(x)=u(t=0,x)=\tanh(x)
\end{cases}.
\end{equation*}

\subsection{{A representation for} finite-size $p=2$ solution {through} Burgers duality}

As mentioned above, the interesting feature in the duality between $p$-spin ferromagnets and the Burgers' hierarchy is that the analysis of the statistical model can be reduced to the study of solutions of linear PDEs for finite $N$ by means of the Cole-Hopf transform. 
{By reversing the duality, we can give a representation of the finite size solution of such spin systems by means of purely PDE methods. In particular, for $p=2$ we can take advantage of the heat-kernel technology to find a description of the Curie-Weiss model by deriving an {\it effective} self-consistency equation for the order parameter. However, we stress that, for these simple systems, the Burgers duality route is not needed, as the solution is easily found even at finite size $N$. Despite this, it is worth to analyze how the whole connection works in both senses.}
\par\medskip
{Recall that the Guerra's generalized partition function is}
\begin{equation}
\label{eq:4.13}
Z_N (t,x)=\sum_{\bb {\sigma}}\exp\Big(-t Nm_N {(\bb\sigma)}^2+N x m_N{(\bb\sigma)}\Big),
\end{equation}
with associated free energy
\begin{equation}
A_N(t,x)= \frac 1N \log Z_N(t,x).
\end{equation}
The spatial derivative $u(t,x)=\partial _x A_N(t,x)$ satisfies the Burgers equation
\begin{equation}
u_t +2 uu_x +\frac1N u_{xx}=0.
\end{equation}
The solution of (interpolated) Curie-Weiss model is equivalent to search the solution of the Burgers equation with initial profile $u_0(x)=u(t=0,x)=\tanh x$. Using the Cole-Hopf transformation $u(t,x)=\frac1N (\log \Psi)_x$, the problem is reduced to the heat equation
\begin{equation}
\label{eq:4.16}
\Psi_t +\frac1N \Psi_{xx}=0.
\end{equation}
By using the previous definitions, we can identify $A_N(t,x)=\frac1N \log \Psi(t,x)$, so that the $\Psi$ function is nothing but the Guerra's generalized partition function \eqref{eq:4.13}.
\par\medskip
The heat equation can be solved via the heat kernel technology, so that the general solution is given by
\begin{equation}
\label{eq:4.17}
\Psi(t,x)= \int dy \, \Psi_0(y) K(t,x-y),
\end{equation}
where $\Psi_0$ is the initial profile of the Cauchy problem, and
\begin{equation}
K(t,x)=\sqrt{-\frac{N}{4\pi t}}\exp\Big(\frac{N x^2}{4 t}\Big).
\end{equation}
\begin{Remark}
We stress that solutions of the form \eqref{eq:4.17} are well-defined for $t<0$, due to the ``wrong'' sign of the temporal derivative in the heat equation \eqref{eq:4.16}. However, this is coherent with our setup, since the link with the thermodynamic model is achieved with $t=-\beta$ with $\beta \in \mathbb R_+$.
\end{Remark}
Now, since $u(t,x)=\frac1N (\log \Psi)_x$ and the initial profile of the solution of Burgers equation is $u_0 (x)= \tanh x$, we immediately have that $\Psi_0 (x)=\cosh ^N (x)$, so that, according to the duality, we have
\begin{equation}
\label{eq:4.19}
A_N (t,x)=\frac1N\log \sqrt{-\frac{N}{4\pi t}} \int dy \cosh^N (y)\exp\Big(\frac{N (x-y)^2}{4 t}\Big),
\end{equation}
and then the argument of the logarithm is an integral representation of the partition function.
In order to make contact with the thermodynamic picture of the Curie-Weiss ferromagnet, we should match the order parameter in terms of the relevant variable on the Burgers side. To do this, we use the fact that $\omega_{t,x} (m_N)=\partial _x A_N (t,x)$, so we take the spatial derivative of the free energy \eqref{eq:4.19}. Then, we have
\begin{equation}
\partial _x A_N(t,x)=\frac{1}{2t} \frac{\int dy { (x-y)} \cosh^N (y)\exp\Big(\frac{N (x-y)^2}{4 t}\Big)}{\int dy \cosh^N (y)\exp\Big(\frac{N (x-y)^2}{4 t}\Big)}.
\end{equation}
By performing the transformation $y=-2t \bar y +x$, we thus have
\begin{equation}
\label{eq:4.21}
\omega_{t,x}(m_N)=\frac{\int d\bar y\ { \bar y} \cosh^N (-2t \bar y +x)\exp({N t \bar y^2})}{\int d\bar y  \cosh^N (-2t \bar y +x)\exp({N t \bar y^2})}.
\end{equation}
Then, it is proved the following 
\begin{Proposition}
The (finite-size) expectation value of the global magnetization for the (interpolated) Curie-Weiss model is equivalent to the first moment of the random variable $\bar y$ distributed according to the probability distribution
\begin{equation}
\label{eq:4.22}
P_{t,x}(\bar y)=C \cosh^N (-2t \bar y +x)\exp({N t \bar y^2}), \quad t<0,
\end{equation}
where $C$ is a normalization constant.
\end{Proposition}

\begin{Remark}
Since the probability distribution \eqref{eq:4.22} is the same associated to the partition function (which in this setup is represented by the $\Psi$ function), we can exactly identify the global magnetization with the random variable $\bar y$.
\end{Remark}

In terms of the $\bar y$ variable, the intensive pressure is \eqref{eq:4.19} is
\begin{equation}
\label{eq:4.23}
A_N (t,x)=\frac1N\log \Big[-t\sqrt{-\frac{N}{\pi t}} \int d\bar y \cosh^N (-2t \bar y+x)\exp({N t \bar y^2}) \Big].
\end{equation}
We can now take the spatial derivative of the intensive pressure. To do this, we consider the identity
$$
\partial _x \cosh^N (-2t \bar y+x)=N\cosh^N (-2t \bar y+x) \tanh (-2t \bar y+x).
$$
Thus, we have
\begin{equation}
\label{eq:4.24}
\omega_{t,x}(m_N)\equiv\partial_x A_N (t,x)=\frac{\int d\bar y \cosh^N (-2t \bar y+x)\tanh(-2t \bar y +x)\exp({N t \bar y^2})}{\int d\bar y \cosh^N (-2t \bar y+x)\exp({N t \bar y^2})} .
\end{equation}

By denoting with $\langle \cdot \rangle_{t,x}$ the average with respect to the probability distribution $P_{t,x}(\bar y)$ and comparing Eqs. \eqref{eq:4.21} and \eqref{eq:4.24}, we arrive at the equality
\begin{equation}
\label{eq:4.28}
\langle \bar y \rangle_{t,x} =\langle  \tanh(-2t \bar y +x) \rangle_{t,x} .
\end{equation}

\begin{Remark}
By setting $t=-\beta$ and $x=0$, we have formally similar self-consistency equations w.r.t. Eq. \eqref{eq:20}, i.e.
$$
\langle \bar y \rangle =\langle  \tanh(2\beta \bar y ) \rangle, 
$$
where $\langle \cdot \rangle= \langle \cdot \rangle_{-\beta,0}$. The substantial difference lies in the fact that in this case the average is performed at finite $N$, so that the r.h.s. is not expressed as a function of the expectation value of the magnetization, since the probability distribution is smooth (or equivalently, it is not peaked on the equilibrium value of the magnetization). Indeed, we checked that the probability distribution \eqref{eq:4.22} at $x=0$ and $t=-\beta$ exhibits the behaviour we expect from the Curie-Weiss picture, see Fig. \ref{fig:pdfs}. In particular, below the critical temperature, the probability distribution displays two different peaks (which are related by the symmetry transformation $y\to -y$) becoming sharper and sharper as $N$ increase, ultimately tending to two Dirac deltas. Above the critical temperature, there is a single peak centered at $y=0$, mimicking the fact that the system has lost its ferromagnetic behavior.
\end{Remark}

The comparison between the explicit solution \eqref{eq:4.15} and the prediction of \eqref{eq:4.21} (for which we checked the equality \eqref{eq:4.28}), is reported in Fig. \ref{fig:numeric}.
	
\begin{figure}[t]
	\centering
	\includegraphics[width=\textwidth]{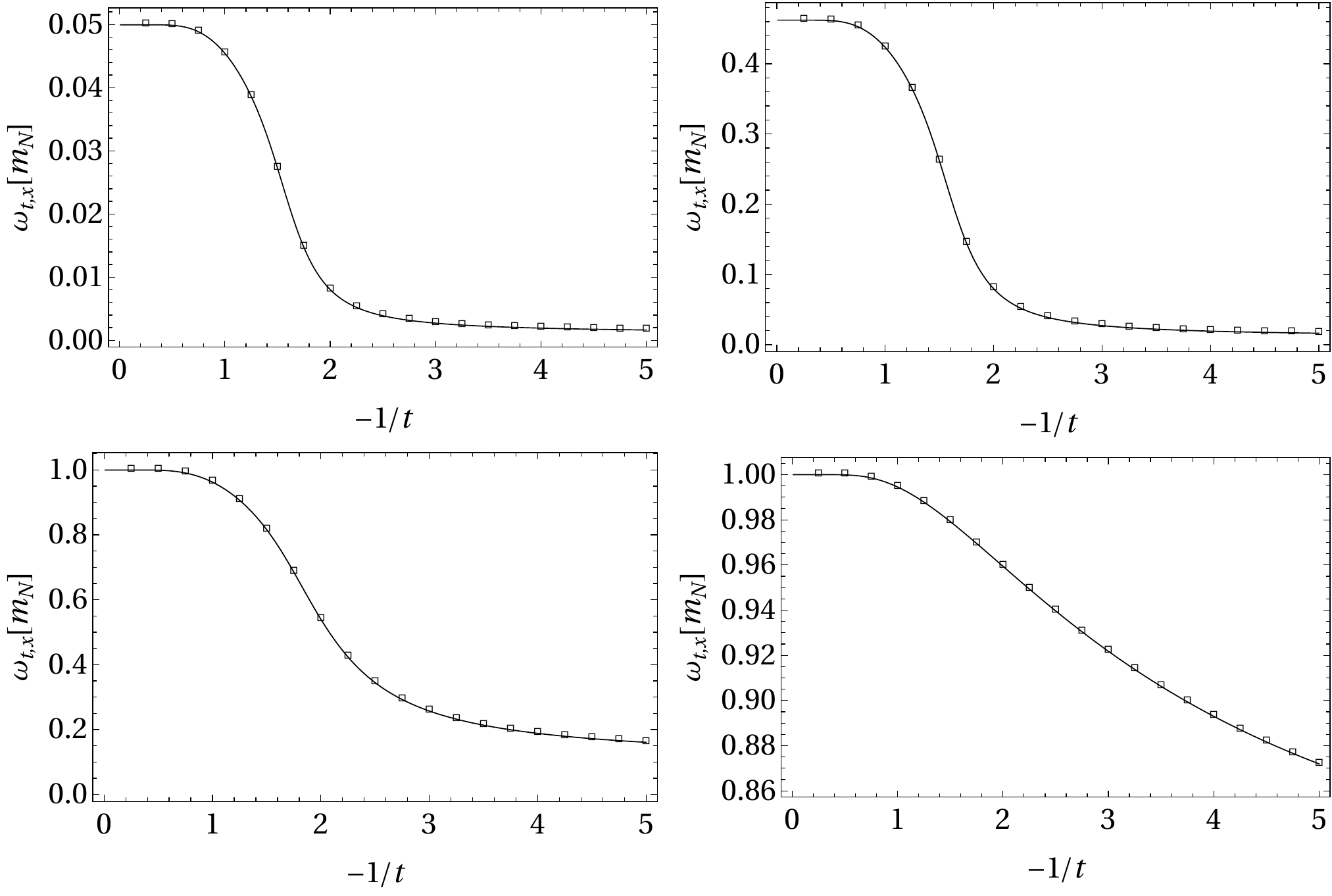}
	\caption{Comparison between the explicit solution \eqref{eq:4.15} (solid curves) and the prediction of \eqref{eq:4.21} (empty squares) for $x=0.001$ (upper left plot), $x=0.01$ (upper right plot), $x=0.1$ (lower left plot) and $x=1$ (lower right plot). In all cases, we fixed the size of the system to $N=50$. The results from the two sides of the duality are in perfect agreement.}\label{fig:numeric}
\end{figure}	
\begin{figure}[t]
\centering
\includegraphics[width=\textwidth]{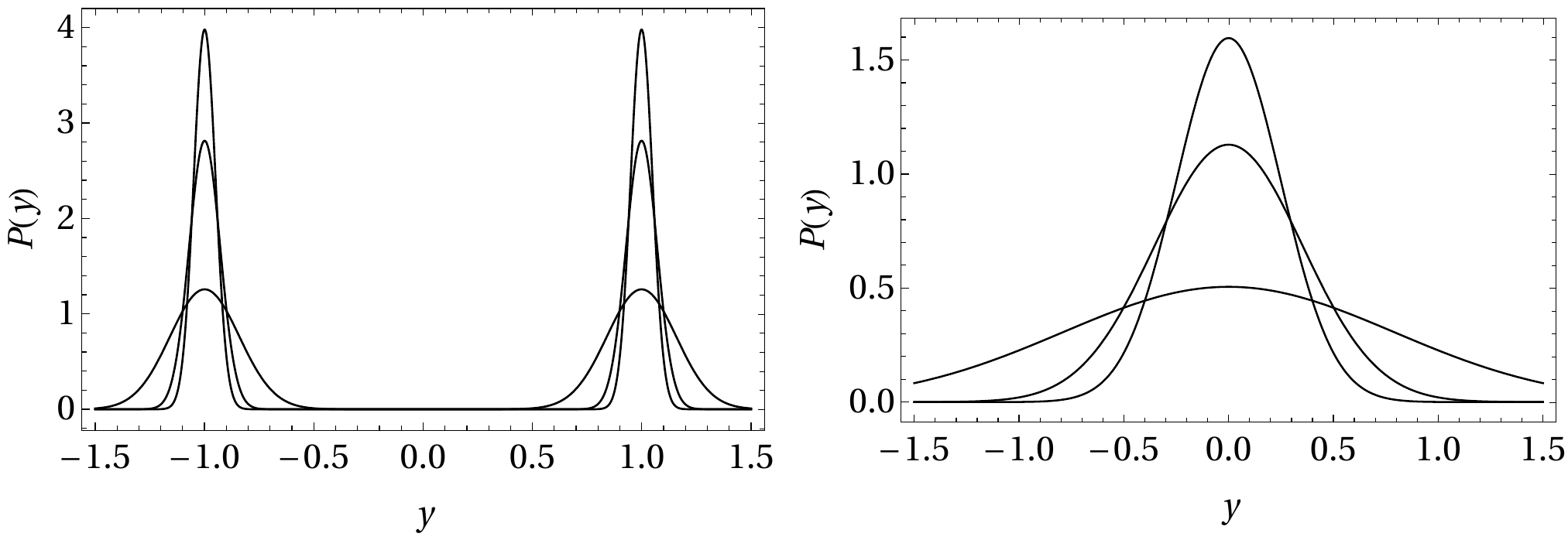}
\caption{Example of probability distributions profile for different values of $\beta=1/T=2$ (left plot, below the critical temperature) and $\beta= 1/T=0.1$ (above the critical temperature). For each plot, we reported different values of $N$: in particular, the three curves correspond (from the broadest to the sharpest) to $N=10$, $N=50$ and $N=100$.}\label{fig:pdfs}
\end{figure}

\section{Conclusions {and further developments}}
Dualities are always powerful tools in mathematical and theoretical physics, since they allow a two-faced investigation of apparently unrelated fields with new mathematical tools, often making easier to derive non-trivial results. In this paper, we examined the relation between the equilibrium dynamics of the $p$-spin ferromagnetic models and the Burgers hierarchy; in particular, the expectation value of the order parameter in the first side is identified with the solution of the initial value profile of the Burgers hierarchy. We also present some examples of application of the duality on both sides. 
{The methods here developed can in principle be applied to other spin models. Indeed, particularly interesting extensions of the present work would be the application of the PDE-statistical mechanics duality to $p$-spin systems with  random external field \cite{lowe,Amaro}, as well as to diluted and finite-connected ferromagnets \cite{dilutedpspin,Franz,Mozeika1,Mozeika}. Even more interesting would be the application of PDE-statistcal mechanics duality to disordered system ({\it i.e.} spin glass) models with interactions of order $p$. In this case, as is easily understood, the situation is much harder than the simple $p$-spin ferromagnets, due to the intrinsic complexity of such systems. Since the structure of pure state is strongly dependent on the realization of the internal disorder and the system size, the structure of Guerra's generalized partition function has to be consistently adapted. For the Sherrington-Kirkpatrick (SK) model, it reads \cite{guerraton3}
$$
Z_N (t,x)=\sum_{\bb \sigma} \exp\Big(\sqrt t \frac{1}{\sqrt{N}}\sum_{1\le i<j\le N}J_{ij}\sigma_i \sigma_j +\sqrt x \sum_{i=1}^N \bar J_i \sigma_i\Big),
$$
and the associated intensive pressure is defined as $A_N (t,x)=\frac1N\mathbb E \log Z_N(t,x)$. Here, both the couplings $J_{ij}$ and the effective external fields $\bar J_i$ are normally distributed, {\it i.e.} $J_{ij}, \bar J_i \sim \mathcal N(0,1)$ for all $i,j=1,\dots,N$, and $\mathbb E$ stands for the average w.r.t. these random variables. The presence of the square root of the ``space-time'' coordinates $t$ and $x$ is motivated by an extensive use of the Wick theorem in the computations, so that the derivatives of the intensive pressure does not explicitly depend on $t$ and $x$. However, for spin-glass models the present methods can only be useful in the thermodynamic limit. This is due to the fact that generally a comprehensive description of spin-glass models is possible in the limit $N\to \infty$, where it is possible to use rigorous methods (such as TLCs theorems or concentration inequalities \cite{guerra3,Tala0}) to achieve the thermodynamic solution. On the technical side, even when working at a replica symmetric level, the mapping of the spin-glass model to mechanical systems explicitly depends on the fluctuation of the order parameter (the replica overlap) w.r.t. its thermodynamic value (see for example \cite{linda} for transport-like PDEs for the SK and Hopfield models), thus in general we can take benefit of the PDE-statistical mechanics duality only in the thermodynamic limit (where fluctuations vanish).
Further, replica symmetry breaking within the Guerra's interpolating framework can be conveniently addressed enriching the interpolation structure. In particular, for the SK model, the $K$-step RSB equilibrium dynamics is achieved in terms of the generalized partition function
$$
Z_K (t,\bb x)=\sum_{\bb \sigma} \exp\Big(\sqrt t \frac{1}{\sqrt{N}}\sum_{1\le i<j\le N}J_{ij}\sigma_i \sigma_j +\sum_{a=1}^K \sqrt {x^{(a)}} \sum_{i=1}^N \bar J_i^{(a)} \sigma_i\Big),
$$
where now $\bb x= (x^{(1)},x^{(2)},\dots , x^{(K)})$ and again $J_{ij}, \bar J_i ^{(a)} \sim \mathcal N(0,1)$ for all $i,j=1,\dots,N$ and $a=1,\dots,K$, see \cite{Barra-Guerra-HJ}. The intensive pressure of the model is now obtained as $A_N (t,x)= \frac1N \mathbb E_0 \log Z_0 (t,x)$, where $Z_{k-1}(t,\bb x)^{\theta_k}=\mathbb E_k (Z_k(t,\bb x)^{\theta_k} )$ for each $k=1,\dots, K$. Here, $\mathbb E_k$ is the expectation value w.r.t. the random variables $\bar J^{(k)}_i$, while $\mathbb E_0$ is the average w.r.t. the coupling realizations $J_{ij}$. Finally, the $\theta _k$ parameters quantify the intensity of each peak in the $N\to\infty$ $K$-RSB overlap distribution \cite{Barra-Guerra-HJ,linda}. This implies that the $K$-RSB approximation of spin-glass equilibrium dynamics can in principle be mapped to $K+1$-dimensional mechanical systems. Also, the application of the whole technology to equilibrium dynamics of neural networks would be highly desirable. To conclude, we strongly believe that PDE-statistical mechanics duality can be, at least in the thermodynamic limit, a very useful tool to investigate properties at the equilibrium of spin-glass models for each interaction order $p$ \cite{Gardner,panchenko4,crisanti_pspin,Kirkpatrick,Kirkpatrick1,Cugliandolo1,Cugliandolo2}. We leave this discussion open for future works.
}

\section{Acknowledgments}
The Author is grateful to Unisalento and Istituto Nazionale di Fisica Nucleare (INFN, Sezione di Lecce) for partial funding. The Author is also very thankful to A. Barra and F. Alemanno for useful discussions and suggestions.

\end{document}